\documentclass[11pt]{article}
\usepackage{fullpage,times,verbatim}
\newcommand{\SHORT}[1]{}
\newcommand{\LONG}[1]{#1}

\usepackage{amsmath, amssymb,amsfonts,amsthm}
\newtheorem{theorem}{Theorem}[section]
\newtheorem{lemma}[theorem]{Lemma}
\newtheorem{corollary}[theorem]{Corollary}

\newcommand{\eps}{\varepsilon}
\newcommand{\e}{\eps}
\newcommand{\TT}{T}  
\newcommand{\IGNORE}[1]{}
\renewcommand{\log}{\lg}
\newcommand{\Patrascu}{P\v{a}tra\c{s}cu}
\newcommand{\Ex}{\mathbb{E}}

\newenvironment{description*}%
  {\vspace{-1ex}\begin{description}%
    \setlength{\itemsep}{-0.5ex}%
    \setlength{\parsep}{0pt}}%
  {\end{description}}
\newenvironment{itemize*}%
  {\vspace{-1ex}\begin{itemize}%
    \setlength{\itemsep}{-0.5ex}%
    \setlength{\parsep}{0pt}}%
  {\end{itemize}}
\newenvironment{enumerate*}%
  {\vspace{-1ex}\begin{enumerate}%
    \setlength{\itemsep}{-0.5ex}%
    \setlength{\parsep}{0pt}}%
  {\end{enumerate}}

\newcommand{\twodots}{\mathinner{\ldotp\ldotp}}

{\makeatletter
 \gdef\xxxmark{%
   \expandafter\ifx\csname @mpargs\endcsname\relax 
     \expandafter\ifx\csname @captype\endcsname\relax 
       \marginpar{xxx}
     \else
       xxx 
     \fi
   \else
     xxx 
   \fi}
 \gdef\xxx{\@ifnextchar[\xxx@lab\xxx@nolab}
 \long\gdef\xxx@lab[#1]#2{{\bf [\xxxmark #2 ---{\sc #1}]}}
 \long\gdef\xxx@nolab#1{{\bf [\xxxmark #1]}}
}

\begin{document}

\title{Orthogonal Range Searching on the RAM, Revisited}
\author{Timothy M. Chan\thanks{This author's work was supported by an
    NSERC grant. School of Computer Science,
             University of Waterloo, tmchan@uwaterloo.ca}
   \and Kasper Green Larsen\thanks{MADALGO, Aarhus University,
     larsen@cs.au.dk.  This author's work was supported in part by
     MADALGO---Center for Massive Data Algorithmics, a Center of the
     Danish National Research Foundation---and in part by a Google 
     Europe Fellowship in Search and Information Retrieval. }
   \and Mihai P\v{a}tra\c{s}cu\thanks{AT\&T Labs, mip@alum.mit.edu}}

\maketitle

\begin{abstract}
We present a number of new results on one of the most extensively
studied topics in computational geometry, orthogonal range searching.
All our results are in the standard word RAM model:
\begin{enumerate}
\item We present two data structures for 2-d orthogonal
  range emptiness. The first achieves $O(n\lg\lg n)$ space and
  $O(\lg\lg n)$ query time, assuming that the $n$ given points are in
  rank space.
This improves the previous results by Alstrup, Brodal, and Rauhe
(FOCS'00), with $O(n\lg^\eps n)$ space and $O(\lg\lg n)$ query time,
or with $O(n\lg\lg n)$ space and $O(\lg^2\lg n)$ query time. Our
second data structure uses $O(n)$ space and answers queries in
$O(\lg^\eps n)$ time. The best previous $O(n)$-space data structure,
due to Nekrich (WADS'07),
answers queries in $O(\lg n/\lg \lg n)$ time.
\item We give a data structure for 3-d orthogonal range reporting
with $O(n\lg^{1+\eps} n)$ space and $O(\lg\lg n + k)$ query time
for points in rank space, for any constant $\eps>0$.
This improves the previous results by Afshani (ESA'08),
Karpinski and Nekrich (COCOON'09), and Chan (SODA'11),
with $O(n\lg^3 n)$ space and $O(\lg\lg n + k)$ query time, or
with $O(n\lg^{1+\eps}n)$ space and $O(\lg^2\lg n + k)$ query time.
Consequently, we obtain improved upper bounds 
for orthogonal range reporting in all constant dimensions above~3.

Our approach also leads to a new data structure for 2-d orthogonal range 
minimum queries with $O(n\lg^\eps n)$ space and $O(\lg\lg n)$ query time
for points in rank space.
\item We give a randomized algorithm for 4-d \emph{offline} dominance range
reporting/emptiness with running time $O(n\log n)$ plus the output size.
This resolves two open problems (both appeared 
in Preparata and Shamos' seminal book):
\begin{enumerate}
\item given a set of $n$ axis-aligned rectangles in the plane,
we can report all $k$ enclosure pairs (i.e., pairs $(r_1,r_2)$ where
rectangle $r_1$ completely encloses rectangle $r_2$)
in $O(n\lg n + k)$ expected time;
\item given a set of $n$ points in 4-d, we can find all maximal points
(points not dominated by any other points)
in $O(n\lg n)$ expected time.
\end{enumerate}
The most recent previous development on (a) was reported back in SoCG'95 by
Gupta, Janardan, Smid, and Dasgupta, whose main result was
an $O([n\lg n + k]\lg\lg n)$ algorithm.  The best previous
result on (b) was an $O(n\lg n\lg\lg n)$ algorithm
due to Gabow, Bentley, and Tarjan---from STOC'84!
As a consequence, we also obtain the current-record time bound for
the maxima problem in all constant dimensions above~4.
\end{enumerate}
\end{abstract}

\SHORT{
\thispagestyle{empty}
\newpage
\setcounter{page}{1}
}

\section{Introduction}
We revisit one of the most fundamental and well-studied classes of problems in
computational geometry, orthogonal range searching. The goal of these
problems is to preprocess a set of $n$ input points in
$d$-dimensional space such that one can efficiently aggregate
information about the points contained in an axis-aligned query
rectangle or box. The most typical types of information computed include
counting the number of points, computing their semigroup or group sum,
determining emptiness, and reporting the points in the query
range. These problems have been
studied extensively for more than three decades, yet many
questions have remained unresolved.
See
e.g.~\cite{arge:indexandrange,indexmodel,subramanian:p-range,afshani:dominance,firstattempt,Brodal00h,McCreight,Bentley.80,Chazelle.functional,Chazelle.filtering.search,Chazelle.Guibas.fractional.I,Chazelle.LB.reporting,Chazelle.Guibas.fractional.II,Chazelle.LB.II,chazelle:offlinerangelb,Lue,Nekrich.SOCG07,Makris.IPL98,Jaja.ISAAC04,Agarwal.Erickson.survey98,Agarwal.survey04,Fredman.LB.semigroup,Willard.LB,patrascu08structures,Nekrich.COCOON}
for just a fraction of the vast amount of publications on orthogonal
range searching.

Recent papers~\cite{AAL,AAL2} have made progress on the pointer machine model and I/O model.
In this paper, we study orthogonal range searching in the standard word RAM model, which
is arguably the most natural and realistic model of computation to consider in internal memory.
We obtain the best RAM upper bounds known to date for a number of
problems, including: 2-d orthogonal range emptiness,
3-d orthogonal range reporting, and offline 4-d dominance range reporting. 

\subsection{Range Searching Data Structures}

In what follows, when stating data structure results, we assume
that all input point sets are in \emph{rank space}, i.e., they have
coordinates on the integer grid $[n]^d=\{0,\dots,n-1\}^d$.
This assumption is for convenience only: in a $w$-bit word RAM 
when all coordinates are in $[U]^d$ with $U=2^w$, we can always reduce to
the rank-space ($U=n$) case by adding to the query time bound a term 
proportional to the cost of predecessor search~\cite{mihai_pred},
which is e.g.~$O(\lg\lg U)$ by van Emde Boas trees~\cite{vEB}
or $O(\log_w n)$ by fusion trees~\cite{FreWil}.
After rank space reduction, all the algorithms mentioned use only
RAM operations on integers of $O(\lg n)$ bits. (The predecessor lower
bound holds even for range emptiness in 2-d, so the additive
predecessor cost in the upper bound is optimal.)

\paragraph{Range reporting in 2-d.}
The most basic version of orthogonal range searching is perhaps range
reporting in 2-d (finding all points
inside a query range).
Textbook description of range trees~\cite{PreShaBOOK} implies a solution with
$O(n\lg n)$ space and $O(\lg n+k)$ query time, where $k$ denotes the output size of the query
(i.e., the number of points reported). Surprisingly, the best
space--query bound for this basic problem is still open.
Chazelle~\cite{Chazelle.functional} gave an $O(n)$-space data
structure with $O(\lg n+ k\lg^\eps n)$ query time, which has been reduced slightly
by Nekrich~\cite{Nekrich:linear} to $O(\lg n/ \lg \lg n+ k\lg^\eps n)$.
(Throughout the paper, $\e>0$
denotes an arbitrarily small constant.)  
Overmars~\cite{overmars} gave a method with $O(n\lg n)$ space and
$O(\lg \lg n + k)$ query time. This query time is optimal
for $O(n\lg^{O(1)}n)$-space structures in the cell probe model
(even for range emptiness in the rank-space case), by reduction from colored predecessor
search~\cite{mihai_pred}.  Alstrup, Brodal and Rauhe~\cite{Brodal00h}
presented two solutions, one achieving $O(n\lg^\e n)$ space
and optimal $O(\lg \lg n+k)$ query time, and one with $O(n \lg \lg n)$
space and $O(\lg^2\lg n +k\lg\lg n)$ query time, both improving Chazelle's
earlier data structures~\cite{Chazelle.functional} with the
corresponding space bounds.  In Section~\ref{sec:2d}, we
present two new solutions.
Our first solution achieves $O(n\lg\lg n)$ space and $O((1+k)\lg\lg n)$
query time, thus strictly improving Alstrup et al.'s second
structure.  Secondly, we present an $O(n)$-space data structure with
query time $O((1+k)\lg^\e n)$, significantly improving the first
term of Nekrich's result.

We can also solve range emptiness in 2-d (testing whether 
a query rectangle contains any input point) by setting $k=0$.
Here, our results are the most attractive, improving on 
all previous results.  For example, our method
with $O(n\lg\lg n)$ space has optimal $O(\lg\lg n)$ query time,
and simultaneously improves both of Alstrup et al.'s solutions
($O(n\lg^\e n)$ space and $O(\lg \lg n)$ time, or $O(n \lg \lg n)$
space and $O(\lg^2\lg n)$ time).

\paragraph{Range reporting in 3-d.}
By a standard reduction, 
Alstrup et al.'s first 2-d result directly implies a data structure
for 3-d orthogonal range reporting with space $O(n \lg^{1+\e}n)$ and
query time $O(\lg n+k)$; this improved an already long chain of previous work.
Nekrich~\cite{Nekrich.SOCG07} was the first to
achieve sublogarithmic query time for 3-d orthogonal range
reporting: his data structure has $O(n \lg^4 n)$ space and
$O(\lg^2 \lg n+k)$ query time. Afshani~\cite{afshani:dominance} subsequently
improved the space to $O(n \lg^3 n)$ while maintaining the same $O(\lg^2 \lg n+k)$ query
time. Karpinski and Nekrich~\cite{Nekrich.COCOON} later reduced the
space to $O(n \lg^{1+\e}n)$ at the cost of increasing the query time
to $O(\lg^3\lg n+k)$, by borrowing ideas
of Alstrup et al.~\cite{Brodal00h}.
In these methods by Afshani~\cite{afshani:dominance}
and Karpinski and Nekrich~\cite{Nekrich.COCOON}, two of the $\lg\lg n$
factors come from orthogonal planar point location.  By using
the most recent result on orthogonal point location by Chan~\cite{chan_pps}, 
one of the $\lg\lg n$ factors can automatically be eliminated in all
of these time bounds.
This still leaves the query time of Karpinski and Nekrich's
structure at $O(\lg^2\lg n+k)$, however.  
In Section~\ref{sec:3d}, we present a new method with
$O(n\lg^{1+\e}n)$ space and optimal $O(\lg\lg n +k)$ query time, simultaneously
improving all previous methods that have linear dependence in $k$.%
\footnote{
As Karpinski and Nekrich~\cite{Nekrich.COCOON} observed,
space can be slightly reduced to $O(n\lg n\lg^{O(1)}\lg n)$
if one is willing to give up linear dependence in $k$, with
query time $O(\lg^2\lg n + k\lg\lg n)$.
}

\paragraph{Range reporting in higher dimensions.}
By a standard reduction, the previous 3-d results
\cite{afshani:dominance,Nekrich.COCOON,chan_pps} imply data structures
for $d$-dimensional orthogonal range reporting with $O(n\lg^d n)$
space and $O(\lg^{d-3}n \lg\lg n + k)$ query time, or
$O(n\lg^{d-2+\e}n)$ space and $O((\lg n/\lg\lg n)^{d-3}\lg^2\lg n +
k)$ query time for $d\ge 4$.  Our result implies a $d$-dimensional
data structure with $O(n\lg^{d-2+\e}n)$ space and $O((\lg n/\lg\lg
n)^{d-3}\lg\lg n + k)$ query time.  This query bound is the best known
among all data structures with $O(n\lg^{O(1)}n)$ space; our space
bound is the best known 
among all data structures with
$O(\lg^{O(1)}n+k)$ query time.

The 4-d case is especially nice, as we get $O(n\lg^{2+\e}n)$ space and
$O(\lg n + k)$ query time.  This query time almost matches \Patrascu's
$\Omega(\lg n/\lg\lg n)$ lower bound~\cite{patrascu08structures} 
for $O(n\lg^{O(1)}n)$-space structures in the cell probe model
for 4-d emptiness.

\paragraph{Range minimum in 2-d.}
Our 3-d range reporting method can also be modified to give a
new result for the 2-d range minimum query problem (see the appendix),
with $O(n\lg^\eps n)$ space and $O(\lg\lg n)$ query time.
\IGNORE{
Our 3-d range reporting method can also be modified to solve the
2-d range minimum query problem (finding the point inside the
query rectangle with the minimum priority, assuming that
each input point is given a priority value).
Our data structure achieves
$O(n\lg^\eps n)$ space and optimal $O(\lg\lg n)$ query time.
(In contrast, modifying Karpinski and Nekrich's 
3-d range emptiness method~\cite{Nekrich.COCOON}
gives $O(n\lg^{O(1)}\lg n)$ space but $O(\lg^2\lg n)$ query time.)
}

\subsection{Offline Range Searching}
Finally, in Section~\ref{sec:offline}, we turn to {\em offline\/} (or
{\em batched\/}) versions of
orthogonal range searching where all queries are given in advance;
the goal is to minimize the total time needed to answer all queries,
including preprocessing.  Offline problems are important,
as efficient algorithms are often obtained through
the use of efficient data structures in offline settings.
Offline problems also raise new challenges, beyond
simply the issue that preprocessing times sometimes get ignored in
analysis of data structures in the literature. Interestingly, the
complexity of offline problems may be fundamentally different from their
online counterparts: examples include predecessor search 
(where the offline problem is related to integer sorting and can be
solved in $O(\sqrt{\lg\lg n})$ expected time per query~\cite{HanTho}),
orthogonal 2-d range counting (where recently Chan and \Patrascu~\cite{ChaPatSODA10}
have obtained an offline $O(\sqrt{\lg n})$ bound
per query, better than the online $O(\lg n/\lg\lg n)$ bound), and
nonorthogonal 2-d point location (where Chan and \Patrascu~\cite{ChaPatSTOC07}
have obtained an offline $2^{O(\sqrt{\lg\lg n})}$ bound, better than
the current online $O(\lg n/\lg\lg n)$ or
$O(\sqrt{\lg U/\lg\lg U})$ bound~\cite{ChaPatFOCS06}).

\paragraph{Offline dominance reporting in 4-d and the rectangle enclosure
problem in 2-d.}  Our main result on offline range searching is a new
algorithm for the offline 4-d dominance reporting problem: given $n$
input points and $n$ query points, report for each query point $q$ all
input points that are dominated by $q$.  Here, $p=(x_1,\ldots,x_d)$ is
{\em dominated\/} by $q=(a_1,\ldots,a_d)$ iff $x_i\le a_i$ for every
$i$, i.e., $p$ lies inside the $d$-sided range
$(-\infty,a_1]\times\cdots\times (-\infty,a_d]$ (an \emph{orthant}).
  In other words, given $n$ red points and $n$ blue points, we want to
  report all pairs $(p,q)$ where the red point $p$ is dominated by the
  blue point $q$.  We give a randomized algorithm that solves this
  problem in $O(n\lg n + k)$ expected time in 4-d, where $k$ denotes
  the total output size. 

(Note that the best known online data structure for 4-d
dominance reporting with $O(\lg n + k)$ query time requires
$O(n\lg^{1+\eps}n)$ space and preprocessing time at least 
as big, and thus is not applicable here.)

In the literature, offline 4-d dominance reporting is studied
under the guise of the 2-d {\em rectangle enclosure\/} problem: given
$n$ rectangles in 2-d, report all pairs $(r_1,r_2)$ where rectangle
$r_1$ completely encloses rectangle $r_2$.  By mapping each rectangle to
a point in 4-d, it is easy to see that the problem reduces to
offline 4-d dominance reporting (in fact, it is equivalent to
 dominance reporting in the ``monochromatic'' case, where we equate the query point set
with the input point set).  This classical problem 
has the distinction
of being the last problem covered in Preparata and Shamos' standard
textbook~\cite{PreShaBOOK}.
In the early 1980s, Vaishnavi and Wood~\cite{VaiWoo}
and Lee and Preparata~\cite{LeePre} both gave $O(n\lg^2 n + k)$-time
algorithms.  The main result of a SoCG'95 paper by
Gupta et al.~\cite{GupSCG95} was an 
$O([n\lg n + k]\lg\lg n)$-time algorithm.
An alternative algorithm
by Lagogiannis et al.~\cite{LaMaTs} obtained the same time bound.
A number of researchers (the earliest seems to be
Bentley and Wood~\cite{BenWoo}) raised
the question of finding an $O(n\lg n +k)$-time algorithm.
Particularly frustrating is the fact that obtaining $O(n\lg n+k)$ time
is easy for the similar-sounding
{\em rectangle intersection\/} problem (reporting all pairs of
intersecting rectangles).
Our new randomized algorithm shows that the 2-d rectangle enclosure
problem can be solved in $O(n\lg n + k)$ time, finally resolving
a 3-decades-old question.

\IGNORE{

Vaishnavi, Wood (1980)
O(n lg^2 n + k) time, O(n lg^2 n)? or O(nlg n)? space
V. Vaishnavi and D. Wood, 
Data structures for the rectangle containment and enclosure problems. Computer 
Graphics and Image Processing 13 (1980), pp. 372–384

Lee, Preparata
O(n lg^2 n + k), O(n) space

papers that mention it as open problem:
J. L. Bentley  and D. Wood.
An optimal worst-case algorithm for reporting intersections of rectangles.
IEEE Trans. Computers, 29:571--577, 1980.
V. Bistiolas, D. Sofotassios, and A. Tsakalidis.
Computing rectangle enclosures. CGTA, 2:303--308, 1993.

another paper with O(n lg nlglg n + klglg n), O(n) space:
G. Lagogiannis, C. Makris, and A. Tsakalidis.
A new algorithm for rectangle enclosure reporting.
IPL, 72:177--182, 1999
}

By a standard reduction, our result implies a randomized $O(n\lg^{d-3}n+k)$-time
algorithm for offline dominance reporting for any constant dimension
$d\ge 4$.

\paragraph{Offline dominance emptiness and the maxima problem.}
Our algorithm can also solve the offline dominance emptiness problem
in $O(n\lg^{d-3} n)$ expected time for any constant $d\ge 4$: 
here, given $n$ input points and $n$ query points, we want to
decide for each query point $q$ whether some input point is dominated by $q$.

A notable application is the {\em maxima\/} problem: given $n$ points,
identify all maximal points, i.e., points that are dominated by no
other point.  Like its cousin, the convex hull problem, this problem
plays a fundamental role in computational geometry and is often used as
examples to illustrate basic algorithmic techniques.  It has many
applications and is related to concepts from other fields
(e.g., skyline queries in databases and Pareto optimality in economics).
The earliest result for dimensions $d\ge 3$
was Kung, Luccio, and Preparata's 
$O(n\lg^{d-2}n)$-time algorithm~\cite{KuLuPr,PreShaBOOK} from 1975.
While progress has been made on probabilistic
results for random point sets~\cite{BeClLe,ClaFOCS94,Gol94}, 
output-sensitive results~\cite{ClaFOCS94,KirSeiSCG85}, and 
even instance-optimal results~\cite{AfBaCh}, 
the best worst-case result for the maxima problem
has remained the one from
Gabow, Bentley, and Tarjan's classic STOC'84 paper~\cite{GaBeTa}.
(That paper is well remembered for introducing Cartesian trees.)
Gabow et al.'s time bound is $O(n\lg^{d-3} n\lg\lg n)$ for $d\ge 4$.
Our (randomized) result implies the first improvement in two and a half decades:
$O(n\lg^{d-3}n)$.  In particular, we obtain the first $O(n\lg n)$
algorithm for the 4-d maxima problem.

\paragraph{Other applications.}
Our offline dominance result also leads to the current
best results for other standard problems (see the appendix), such as bichromatic $L_\infty$-closest pair
and $L_\infty$-minimum spanning tree for $d\ge 4$.

\paragraph{Organization.}
In the three subsequent sections, we describe our new methods for 2-d
orthogonal range reporting, 3-d orthogonal range reporting, and
offline 4-d dominance range reporting.  These sections are independent
of each other and can be read separately.  Interestingly, our
techniques for 2-d range reporting are not based on Alstrup et al.'s
previous grid-based approach~\cite{Brodal00h} but draw on new ideas
related to succinct data structures. Our 3-d range reporting structure
is based on Alstrup et al.'s approach, but with new twists.  Finally,
our 4-d offline algorithm involves an unusual (and highly nonobvious)
mixture of bit-packing techniques~\cite{ChaPatSODA10} and classical
computational geometric tools (Clarkson--Shor-style random sampling).

\section{Range Reporting in 2-d}\label{sec:2d}

\newcommand{\occ}{k}

The goal of this section is to prove:
\begin{theorem}  \label{thm:2d}
For any $2 \le B \le \lg^\eps n$, we can solve 2-d orthogonal range
reporting in rank space with: $O(n \cdot B \lg\lg n)$ space and 
$O(\lg\lg n + k \cdot \log_B \lg n)$ query time; or $O(n \cdot \log_B
\lg n)$ space and $(1+k) \cdot O(B \lg\lg n)$ query time.
\end{theorem}

At the inflection point of the two trade-offs, we get a data structure
with space $O(n \lg\lg n)$ and query time $(1+k) \cdot O(\lg\lg
n)$. At one extreme point, we have space $O(n)$ and query time $(1+k)
\cdot O(\lg^\eps n)$. At the other extreme, we have space $O(n
\lg^\eps n)$ and query time $O(\lg\lg n + k)$, thus matching the
bounds of Alstrup et al. Note that all these tradeoffs also apply to
emptiness with $k$ set to $0$. This can be seen through a black-box
reduction: Assume a reporting data structure with query time $t_1+t_2
k$ is available. Given an emptiness query, run the query algorithm on
the reporting data structure using the same query. If the query
algorithm terminates within $t_1$ computation steps, we immediately
get the answer, otherwise we terminate after $t_1+1$ operations, at
which point we know $k>0$ and thus we know the range is nonempty.

We will describe a linear-space reduction from 2-d orthogonal range
reporting to the following ``ball inheritance'' problem, a
pointer-chasing-type problem reminiscent of fractional
cascading~\cite{Chazelle.Guibas.fractional.I}. Consider a perfect binary tree with $n$ leaves. Also
consider $n$ labelled balls, which appear in an ordered list at the
root of the tree. We imagine distributing the balls from the root down
to the leaves, in $\log n$ steps. In the $i$-th step, a node on level
$\lg n-i$ contains a subset of the balls in an ordered list, where the
order is the same as the original order at the root. Each ball chooses
one of the two children of the node and is ``inherited'' by that
child. The number of balls in each node across a level is the
same. That is, on level $i$, each node contains exactly $2^i$ balls,
and each leaf contains exactly one ball.

Given the inheritance data described above, the goal is to build a
data structure that answers the following type of query: given a node
and an index into its list of balls, what leaf does the indicated ball
eventually reach? We may imagine each ball as having $\lg n$ copies,
one at each node on its root-to-leaf path. Conceptually, each ball
stores a pointer to its copy on the level below. The identity of a
ball on level $i$ consists of a node at level $i$, and the index of
the ball in the node's list. The goal is, given (the identity of) a
ball on some level, to traverse the pointers down to the tail of the
list, and report the leaf's identity.

In Section~\ref{sec:2d-ptr}, we give space/time trade-offs for the
problem with results parallel to Theorem~\ref{thm:2d}. These
trade-offs come out naturally given the definition of ball
inheritance: our data structure mimicks skip lists on $n$ independent
lists with the copies of each ball. In Section~\ref{sec:2d-red}, we
give a reduction from 2-dimensional range reporting to this abstract
ball-inheritance problem.

\subsection{Solving the Ball-Inheritance Problem} \label{sec:2d-ptr}

This subsection will prove:
\begin{lemma}
\label{lem:ballinherit}
For any $2 \le B \le \lg^\eps n$, we can solve the ball-inheritance
problem with: (1) space $O(n B\lg\lg n)$ and query time $O(\log_B \lg
n)$; or (2) space $O(n \log_B \lg n)$ and query time $O(B \lg\lg
n)$. 
\end{lemma}

Using standard techniques, one can represent the pointers on each
level of the tree with $O(n)$ bits such that we can traverse a pointer
in constant time. This uses the rank problem from succinct data
structures: represent a bit vector $A[1\twodots n]$ using $O(n)$ bits,
to answer $\mathrm{rank}(k) = \sum_{i \le k} A[i]$.  In fact solutions
with very close to $n$ bits of space and constant query time are known
\cite{patrascu08succinct}.  For every node, we store a solution to the
rank problem among its balls, where ``0'' denotes a ball going to the
left child, and ``1'' a ball going to the right child. The index of a
ball in the right child is $\mathrm{rank}(i)$ evaluated at the
parent. The index of ball $i$ in the left child is $i -
\mathrm{rank}(i)$. 

This trivial data structure uses $O(n\lg n)$ bits in total, or $O(n)$
words, but has $O(\lg n)$ query time. For faster queries, the query
will need to skip many levels at once. We use an easy generalization
of rank queries, which we prove in the appendix:
\begin{lemma}   \label{lem:succinct}
Consider an array $A[1\twodots n]$ with elements from some alphabet
$\Sigma$. We can construct a data structure of $O(n\lg \Sigma)$ bits
which answers in constant time $\mathrm{rank}(k) = $ the number of
elements in $A[1 \twodots k]$ equal to $A[k]$.
\end{lemma}

In our context, the lemma implies that we can store all pointers from
balls at level $i$ to balls at level $i+\Delta$ using $O(n\Delta)$
bits of space. Indeed, each ball can be inherited by $2^\Delta$
descendants of its current node (the alphabet $\Sigma$ will denote
this choice of the descendant $\Delta$ levels below). To compute the
index of a ball in the list of its node at level $i+\Delta$, we simply
have to count how many balls before it at level $i$ go to the same
descendant (a rank query).

For intuition of how to use this building block, consider an abstract
problem. We need to augment a linked list of $m$ nodes in a manner
similar to a skip list. Any node is allowed to store a pointer
$\Delta$ nodes ahead, but this has a cost of $\Delta$. The goal is to
reach the tail of the list from anywhere in a minimal number of hops,
subject to a bound on the total cost of the skip pointers. For any $2
\le B \le m$, we can solve this problem as follows:
\begin{itemize*}
\item Traversal time $O(\log_B m)$ with pointer cost $O(m \cdot B
  \log_B m)$. Define the level of a node to be the number of trailing
  zeros when writing the node's position in base $B$. Each node on
  level $i$ stores a pointer to the next node on level $i+1$, or to
  the tail if no such node exists. The cost of a node at level $i$ is
  $O(B^{i+1})$, which is $O(\frac{m}{B^i} B^{i+1}) = O(mB)$ across a
  level. The traversal needs to look at $O(\log_B m)$ pointers (one
  per level) before reaching the tail.

\item Traversal time $O(B \log_B m)$ with pointer cost $O(m \cdot
  \log_B m)$.  Each node on level $i$ stores a pointer that skips
  $B^i$ nodes, or to the tail if no such node exists. In other words,
  each node on level $i$ stores a pointer to the next node on level
  $i$ or higher (whichever comes first). The cost of a node on level
  $i$ is $O(B^i)$, so all nodes on level $i$ cost $O(\frac{m}{B^i}
  B^i) = O(m)$. The total cost is thus $O(m \log_B m)$. We can reach
  the tail from anywhere with $O(B \log_B m)$ pointer traversals,
  since we need at most $B$ nodes on each level, before reaching a
  node on a higher level.
\end{itemize*}

Returning to the ball-inheritance problem, we will implement the above
strategies on the $n$ lists of copies of each ball, using
Lemma~\ref{lem:succinct} to store pointers. We use the first strategy
in the regime of fast query time, but higher space (tradeoff (1) in
Lemma~\ref{lem:ballinherit}). Nodes on levels of the tree that are a
multiple of $B^i$ store pointers to the next level multiple of
$B^{i+1}$. This costs $O(B^{i+1})$ bits per ball, so the total cost is
$\sum_i \frac{\lg n}{B^i} \cdot O(B^{i+1}) = O(\lg n \cdot B \cdot
\log_B \lg n)$ bits per ball. This is $O(n B \log_B \lg n)$
words of space. The query time is $O(\log_B \lg n)$, since in each step, we
jump from a level multiple of $B^i$ to a multiple of $B^{i+1}$. Since
the bound is insensitive to polynomial changes in $B$, the trade-off
can be rewritten as: space $O(n B \lg\lg n)$ and query time $O(\log_B
\lg n)$.

The second strategy gives low space, but slower query, i.e.~tradeoff
(2) in Lemma~\ref{lem:ballinherit}. Nodes on levels that are a
multiple of $B^i$ store pointers to $B^i$ levels below (or to the
leaves, if no such level exists). The cost of such a level is $O(B^i)$
bits per ball, so the total cost is $\sum_i \frac{\lg n}{B^i} \cdot
O(B^i) = O(\lg n \cdot \log_B \lg n)$ bits per ball. This is space
$O(n \log_B \lg n)$ words. The query time is $O(B \log_B \lg n)$,
since we need to traverse at most $B$ levels that are multiples of
$B^i$ before reaching a level multiple of $B^{i+1}$. Thus, we obtain
query time $O(B \lg\lg n)$ with space $O(n \log_B \lg n)$.

\subsection{Solving Range Reporting} \label{sec:2d-red}

This subsection will show:

\begin{lemma}
\label{lem:balltoreport}
If the ball inheritance problem can be solved with space $S$ and query
time $\tau$, 2-d range reporting can be solved with space $O(S+n)$ and
query time $O\big(\lg\lg n + (1 + k) \cdot \tau \big).$
\end{lemma}

Consider $n$ points in 2-d rank space; we may assume $n$ to be a power
of two. We build a perfect binary tree over the $x$-axis. Each ball
will represent a point, and the leaf where the ball ends up
corresponds to its $x$ coordinate. The order of balls at the root is
the sorted order by $y$ coordinate. We store a structure for this ball
inheritance problem. The true identity of the points (their $x$ and
$y$ coordinates) are only stored at the leaves, taking linear
space. We will now describe additional data structures that allow us
to answer range reporting with query time $O(\lg\lg n)$, plus
$(O(1)+k)$ ball inheritance queries.

Our first ingredient is a succinct data structure for the range
minimum problem (RMQ). Consider an array $A$ of $n$ keys (which can be
compared). The query is, given an interval $[i,j]$, report the index
of the minimum key among $A[i], \dots, A[j]$. Note that a data
structure for this problem does not need to remember the
keys. Information theoretically, the answer is determined if we know
the Cartesian tree~\cite{GaBeTa} of the input; a tree takes just $2n$ bits to
describe. Effective data structures matching this optimal space bound
are known. See, for example, \cite{fischer10rmq}, which describes a
data structure with $2n + o(n)$ bits of space and $O(1)$ query time.

In each node that is the right child of its parent, we build a
succinct RMQ data structure on the points stored in the subtree rooted
at that node. In this structure, we use the $y$ rank of the points as
indices in the array, and their $x$ coordinates as keys. In each node
that is a left child of its parent, we build a range-maximum data
structure (equivalently, an RMQ data structure on the mirrored input).
Since each data structure takes a number of bits linear in the size of
the node, they occupy a total of $O(n\lg n)$ bits, i.e. $O(n)$ words
of space.

To report points in the range $[x_1, x_2] \times [y_1, y_2]$, we
proceed as follows:
\begin{enumerate*}
\item Compute the lowest common ancestor $\mathrm{LCA}(x_1, x_2)$ in
  the perfect binary tree. This is a constant time operation based on
  the xor of $x_1$ and $x_2$: the number of zero bits at the end
  indicates the height of the node, and the rest of the bits indicate
  the nodes identity. (For instance, we can use an array of $n$
  entries to map the $x_1 \oplus x_2$ to the right node.)

\item We convert $[y_1, y_2]$ into the rank space of points inside the
  left and right child of $\mathrm{LCA}(x_1, x_2)$. This entails
  finding the successor $\hat{y}_1$ of $y_1$ among the $y$ values of
  the points under the two nodes, and the predecessor $\hat{y}_2$ of
  $y_2$. See below for how this is done.

\item We descend to the right child of $\mathrm{LCA}(x_1, x_2)$. Using
  the RMQ data structure, we obtain the index $m$ (the $y$ rank) of
  the $x$-minimum point in the range $[\hat{y}_1, \hat{y}_2]$. We use
  the ball-inheritance structure to find the leaf of this point. We
  retrieve the $x$ coordinate of the point from the leaf, and compare
  it to $x_2$. If greater, there is no output point in the right node.
  If smaller, we report this point and recurse in $[\hat{y}_1, m-1]$
  and $[m+1, \hat{y}_2]$ to report more points. We finally apply the
  symmetric algorithm in the left child of $\mathrm{LCA}(x_1, x_2)$,
  using the range maxima until the points go below $x_1$.
\end{enumerate*}

\noindent
The cost of step 3 is dominated by the queries to the ball-inheritance
problem. The number of queries is two if the range is empty, and
otherwise at most twice the number of points reported in each child of
the $\mathrm{LCA}$.

We now describe how to support step 2 in $O(\tau + \lg\lg n)$ time,
with just $O(n)$ space in total. We will use a succinct index for
predecessor search. Consider supporting predecessor queries in a
sorted array $A[1\twodots n]$ of $w$-bit integers. If we allow the
query to access entries of the array through an oracle, it is possible
to obtain a data structure of sublinear size. More precisely, one can
build a data structure of $O(n \lg w)$ bits, which supports queries in
$O(\lg w)$ time plus oracle access to $O(1)$ entries; see \cite[Lemma
  3.3]{grossi09succinct}. (This idea is implicit in fusion trees
\cite{fredman93fusion}, and dates further back to
\cite{ajtai84hashing}.)

We build such a data structure on the list of $y$ coordinates at each
node. The oracle access to the original $y$ coordinates is precisely
what the ball-inheritance problem supports. Since our points have
coordinates of $O(\lg n)$ bits, each data structure uses $O(\lg\lg n)$
bits per point, so the total space is $O(n \lg\lg n)$ words. To reduce
the space to linear, we store these predecessor structures only at
levels that are a multiple of $\lg\lg n$. From $\mathrm{LCA}(x_1,
x_2)$, we go up to the closest ancestor with a predecessor
structure. We run predecessor queries for $y_1$ and $y_2$ at that
node, which take $O(\lg\lg n)$ time, plus $O(1)$ queries to the
ball-inheritance problem. Then, we translate these predecessors into
the rank space of the left and right child of $\mathrm{LCA}(x_1,
x_2)$, by walking down at most $\lg\lg n$ levels in the
ball-inheritance problem (with constant time per level). This
concludes the proof of Lemma~\ref{lem:balltoreport}, which combined
with Lemma~\ref{lem:ballinherit} proves Theorem~\ref{thm:2d}.

\section{Range Reporting in 3-d}\label{sec:3d}

In this section, we present a new data structure for 3-d orthogonal
range reporting.  We find it more convenient now to ignore the 
default assumption that points are given in rank space.
The special case of dominance (i.e., 3-sided)
reporting can already be solved with
$O(n)$ space and $O(\lg\lg U + k)$ time 
by known methods~\cite{afshani:dominance}, using the latest
result on orthogonal planar point location~\cite{chan_pps}.
We will show how to ``add sides''  without changing the asymptotic query time
and without increasing space by too much.

\subsection{The 3-d 4-Sided Problem}

Our method is based on a simple variant of Alstrup, Brodal, and Rauhe's 
grid-based method~\cite{Brodal00h}.   
Instead of a grid of dimension near $\sqrt{n}\times\sqrt{n}$
as used by Alstrup et al.'s and Karpinski and Nekrich's method~\cite{Nekrich.COCOON}, 
our key idea is to use a grid of dimension near $(n/t)\times t$ 
for a judiciously chosen parameter~$t$.

Specifically, consider the problem of answering
range reporting queries for 4-sided boxes in 3-d, i.e.,
the boxes are bounded in 4 out of the 6 sides where the unbounded
sides are from different coordinate axes.
W.lo.g., assume that query ranges are
unbounded from below in the $y$ and $z$ directions.
Suppose that there is a base data structure for solving the 3-d 4-sided problem
with $S_0(n)$ space {\em in bits\/} and $Q_0(n,k)$ query time.

\paragraph{The data structure.}
Fix parameters $t$ and $C$ to be set later.  
Let $S$ be a set of $n$ points in $[U]^3$.
Build a 2-d grid on the $xy$-plane
 with $n/(Ct)$ rows and $t$ columns, so that
each row contains $Ct$ points and each column contains
$n/t$ points.\footnote{Abusing notation slightly,
we will not distinguish between a 2-d region (e.g., a row, column, or grid cell)
and its lifting in 3-d.
For simplicity, we ignore floors and ceilings.}
The number of grid cells is $n/C$.
\LONG{\par}
We build a data structure for $S$ as follows:
\begin{enumerate}
\item[0.] For each of the $t$ columns (in left-to-right order), 
build a data structure recursively
for the points inside the column.
\item For each of the $n/(Ct)$ rows, build a base data structure directly for
the points inside the row, with $S_0(Ct)$ space and
$Q_0(Ct,k)$ query time.  Also build a predecessor search structure 
for the $n/(Ct)$ horizontal grid lines.
\item For each of the $t$ columns, build a 3-sided reporting
data structure~\cite{chan_pps} for the points inside the column,
with $O(n/t)$ words of space, or
$O((n/t)\lg U)$ bits of space, and $O(\lg\lg U + k)$ query time.
\item Let $G$ be the set of at most $n/C$ points formed by
taking the $z$-lowest point out of each nonempty grid cell.  Build 
a data structure for $G$ for 4-sided queries using any known method~\cite{chan_pps}
with $O((n/C)\lg^{O(1)}n)$ space and $O(\lg\lg U + k)$ query time.
\item Finally, for each nonempty grid cell, store the
list of all its points sorted in increasing $z$-order.
\end{enumerate}
Note that by unfolding the recursion, we can view our 
data structure as a degree-$t$ tree $\TT$ where
the points in the leaves are arranged in $x$-order and
each node stores various auxiliary data structures (items 1--4).

The space usage in bits satisfies the recurrence
\[ S(n)\ =\ t S(n/t) + (n/(Ct))S_0(Ct) + O(n\lg U + (n/C)\lg^{O(1)}n).
\]
For the base case, we have $S(n)=O(S_0(Ct))$ for $n<Ct$.
Solving the recurrence gives 
\[ S(n)\ =\ O(\log_t n \cdot [(n/(Ct))S_0(Ct) + n\lg U + (n/C)\lg^{O(1)}n])
\]
(assuming that $S_0(n)/n$ is nondecreasing).
The third term disappears by setting $C=\lg^c n$ for a sufficiently
large constant~$c$.


\paragraph{The query algorithm.}
Suppose we are
given a query range $q=[x_L,x_R]\times (-\infty,y_0]\times (-\infty,z_0]$
and the $x$-ranks of $x_L$ and $x_R$ w.r.t.\ the input point set.
Let $v$ be the lowest common ancestor
of the two leaves of $\TT$ whose $x$-range contain $x_L$ and $x_R$. 
We can find $v$ by performing a word operation on the two given $x$-ranks 
(no special LCA data structures are required since $\TT$ is perfectly balanced).

From now on, we work exclusively at node $v$ of the tree.
There, $q$ intersects more than one column.
Say $x_L$ and $x_R$ are in columns $j_L$ and $j_R$ (which can
be identified in $O(1)$ time as we know the $x$-ranks).
Say $y_0$ is in row $i$, computable by predecessor 
search in $O(\lg\lg U)$ time.
We can then answer the query as follows:
\begin{enumerate}
\item Let $q_T$ be the (``top'') portion of $q$ inside row $i$.
Report all points in $q_T$ (which is 4-sided) by the base data structure 
at row $i$.  The cost is $Q_0(Ct,k')$ if $k'$ denotes the number of points in $q_T$.
\item Let $q_L$ and $q_R$ be the portions of $q-q_T$ inside
columns $j_L$ and $j_R$ respectively.
Report all points in $q_L$ and $q_R$ (which are 3-sided inside
columns $j_L$ and $j_R$ respectively) by
the 3-sided data structure at columns $j_L$ and $j_R$.  The cost is
$O(\lg\lg U)$ plus the number of points in $q_L$ and $q_R$.
\item Let $q_I$ be the remaining (``interior'') portion $q-(q_T\cup q_L\cup q_R)$.
Find all points of $G$ in $q$ by
querying the data structure for $G$.
The cost is at most $O(\lg\lg U)$ plus the number of points in $q_I$.
\item For each point $s\in G$ found in step~3, report all points 
in $s$'s grid cell with $z$-values below $z_0$ by a linear search 
over the cell's $z$-sorted list.
The cost is linear in the number of points in $q_I$.
\end{enumerate}
The overall query time  is thus
$Q(n,k)\ =\ Q_0(Ct,k') + O(\lg\lg U + k-k')$
for some $k'\le k$.

\paragraph{Bootstrapping.}
Assume the availability of a solution with
$S_0(n)=O(n\lg U + n\lg^{1+1/\ell}n)$ and
$Q_0(n,k)=O(\lg\lg U + k)$ for a constant $\ell$.
(For a base case with $\ell\in(0,1]$, we can start with any known method with
$O(n\lg^{O(1)}n)$ space and $O(\lg\lg U+k)$ query time~\cite{chan_pps}.)
Setting $t$ with $\lg t=\lg^{\ell/(\ell+1)}n$ then yields 
\[ S(n) \ =\ O((\lg n)/(\lg t) \cdot [n\lg U + n\lg^{1+1/\ell}Ct])
        \ =\ O(n\lg U\lg^{1/(\ell+1)}n)
\]
and $Q(n,k)= O(\lg\lg U + k)$.

We need one last trick: rank space reduction.  Initially, 
store the $x$-, and $y$-, and $z$-values in sorted arrays,
build predecessor search structures for them, and
afterwards, replace all values by their ranks.  This way, we have
reduced $U$ to $n$, and the space (in bits) of the data structure improves to
$O(n\lg U + n\lg^{1+1/(\ell+1)}n)$.  For a query range $q$, we can initially
determine the $x$-, $y$-, and $z$-ranks of $q$'s endpoints in 
$O(\lg\lg U)$ time~\cite{vEB} before running the query algorithm.  Incidentally,
this also fulfills
the assumption that the $x$-ranks of $q$'s endpoints are given. 
After the query, we can recover the $x$-, $y$-, and $z$-values of
each reported point by looking up the sorted arrays. The query time 
remains $O(\lg\lg U+k)$
(though the constant factor in the $k$ term increases).

By bootstrapping $\lceil 1/\eps\rceil$ times, we finally obtain a solution
for the 3-d 4-sided problem with $O(n\lg U + n\lg^{1+\eps}n)$ bits of
space, i.e., 
$O(n\lg^\eps n)$ words of space (by packing), and $O(\lg\lg U + k)$ query time.

\subsection{The 3-d 5-Sided/6-Sided Problem}

\SHORT{
We can solve the 3-d 5-sided problem in the same way, and
we can reduce the 3-d general (i.e., 6-sided) problem to the 5-sided
case by a standard reduction by paying a $\log$ factor in space
(See the appendix). Therefore:
}

\LONG{
We can solve the 3-d 5-sided problem in the same way.  W.l.o.g., assume that
the ranges are unbounded from below in the $z$ direction.  Item~2 
now stores 4-sided data structures, and space increases
by a $\log^\eps n$ factor only as a result.  The query algorithm proceeds
similarly.  We now have an additional bottom portion $q_B$,
but $q_T,q_B,q_L,q_R$ are all 4-sided, unless $q$ lies completely
inside a column (in which case we only need one query to
a base data structure).
Our method thus solves the 3-d 5-sided problem with
$O(n\lg^{O(\eps)} n)$ words of space and $O(\lg\lg U + k)$ query time.

It is known (e.g., see~\cite{Nekrich.SOCG07})
that the $j$-sided problem can be reduced to the $(j-1)$-sided
problem by standard binary divide-and-conquer, 
where the space increases by a logarithmic factor but
the query time is unchanged (if it exceeds $\lg\lg$).
Thus, we can get the following result for
the 3-d general (i.e., 6-sided) problem:
}

\begin{theorem}\label{thm:3d}
There is a data structure for 3-d orthogonal range reporting
with $O(n\lg^{1+\eps}n)$ space (in words) and $O(\lg\lg U + k)$ query time.
\end{theorem}

\newcommand{\THREEDAPPENDIX}{



\subsection{Higher dimensions and applications.}
It is known that $d$-dimensional orthogonal range reporting can be
reduced to $(d-1)$-dimensional orthogonal range reporting by 
using a range tree with fan-out $b$,
where the space increases by a $b^{O(1)}\log_b n$ factor and
the query time increases by a $\log_b n$ factor.  By setting $b=\lg^\eps n$
(and applying rank space reduction at the beginning), Theorem~\ref{thm:3d} implies:

\begin{corollary}
There is a data structure for $d$-dimensional orthogonal range reporting
for any constant $d\ge 4$ with
$O(n\lg^{d-2+\eps}n)$ space and $O((\lg n/\lg\lg n)^{d-3}\lg\lg n + k)$
query time.
\end{corollary}

Our method also works for emptiness queries; the same bounds hold
with $k$ set to 0. Here, dominance emptiness structures 
are sufficient in item~2, and item~4 is unnecessary.

\IGNORE{
Our 3-d range reporting method can also be modified to solve the
2-d range minimum query problem (finding the point inside the
query rectangle with the minimum priority, assuming that
each input point is given a priority value).
Our data structure achieves
$O(n\lg^\eps n)$ space and optimal $O(\lg\lg n)$ query time.
(In contrast, modifying Karpinski and Nekrich's 
3-d range emptiness method~\cite{Nekrich.COCOON}
gives $O(n\lg^{O(1)}\lg n)$ space but $O(\lg^2\lg n)$ query time.)
}

Range minimum queries  (finding the point inside a
query range with the minimum priority, assuming that
each input point is given a priority value) are closely related.
For example, the decision version of 2-d range minimum queries (deciding whether
the minimum is at most a given value) reduces to 3-d 5-sided emptiness
queries. It is no surprise then that we can obtain
the same result for 2-d range minimum queries as 3-d 5-sided emptiness.
Here, in item~2, we need to replace 3-d 3-sided emptiness structures
with 2-d dominance range minimum structures, but 2-d dominance range
minimum reduces to point location in the vertical projection of a
lower envelope of 3-d orthants.  This is an orthogonal 2-d point location problem,
which can be solved with $O(n\lg U)$ space in bits and $O(\lg\lg U)$ query 
time~\cite{chan_pps}.  The same analysis thus carries through.

\begin{theorem}
There is a data structure for 2-d range minimum queries
with $O(n\lg^{\eps}n)$ space and $O(\lg\lg U)$ query time.
\end{theorem}

In contrast, modifying Karpinski and Nekrich's 
3-d range emptiness method~\cite{Nekrich.COCOON}
yields a data structure for 2-d range minimum queries with
$O(n\lg^{O(1)}\lg n)$ space but $O(\lg^2\lg n)$ query time.

Our method can also give an alternative solution
to 2-d orthogonal range reporting with $O(n\lg^\eps n)$ space
and $O(\lg\lg U + k)$ time.  This solution is arguably slightly simpler 
than Alstrup, Brodal, and Rauhe's original method~\cite{Brodal00h}, as
their method requires constant-time 1-d range queries as a subroutine,
and also requires a more intricate way of handling rank space reduction.
}
\LONG{\THREEDAPPENDIX}


\section{Offline Range Reporting}\label{sec:offline}
In this section, we present our $O(n\lg n + k)$ expected time solution
for the offline 4-d dominance reporting problem on $n$ query points and
$n$ input points, where $k$ denotes the total output size of all $n$
queries. In this section, we fix $w=\eps\log N$ where $N$ denotes the
maximum input size. Any $w$-bit word operation we introduce can 
be simulated in $O(1)$ time by table lookup, after preprocessing
in sublinear time $2^{O(w)}=N^{O(\eps)}$.

\subsection{Preliminaries}

We begin by describing some key subroutines and tools we need. 
The first subroutine is an algorithm for a special case of
offline 2-d orthogonal point location.  
Chan and \Patrascu~\cite{ChaPatSODA10} recently studied
the offline 2-d orthogonal range counting problem and obtained a linear-time
algorithm for the case when the number of points is smaller than
$2^{O(\sqrt{w})}$.  From this result, they then obtained an 
$O(n\sqrt{\lg n})$ algorithm for
the general case.  We apply their bit-packing technique and show
that a similar result holds for orthogonal point location
(see the appendix for the proof):

\begin{lemma}
\label{lem:offpl}
There is an algorithm for offline 2-d orthogonal point location 
on $n$ query
points and $n$ disjoint axis-aligned rectangles that runs in time
$O(n)$ if $n\le 2^{O(\sqrt{w})}$ and the coordinates have been pre-sorted.
\end{lemma}

The second subroutine is a preliminary method for the offline $d$-dimensional
orthogonal range reporting problem.
A straightforward $b$-ary version of the range tree, combined with
a trivial method for the 1-d base case, easily gives the following
bound, which with the right choice of $b$ will turn out to be crucial in
establishing our 4-d result:

\begin{lemma}
\label{lem:offstab}
There is an algorithm for offline $d$-dimensional orthogonal
range reporting on $n$ points and $m$ boxes that runs in time 
$O(n\log^{d-1}_b n+b^{d-1} m \log^{d-1}_b n+k)$, where $b\geq 2$ is a
parameter and $k$ is the total output size, if coordinates have been
pre-sorted.
\end{lemma}

\newcommand{\STAIR}{{\cal P}}
\newcommand{\VD}{{\cal VD}}
\newcommand{\D}{\Delta}
The main geometric tool we use is a randomized version of 
{\em shallow cuttings\/} in 3-d.  Let $S$ be a set of $n$ points in 3-d.
Pick a random sample $R$ where each point of $S$ is included independently
with probability $1/K$ for a fixed parameter $K$.  
Define the {\em staircase polyhedron\/}
$\STAIR(R)$ to be the lower envelope of the orthants 
$O_s=[x,\infty) \times [y,\infty) \times [z,\infty)$ over all the
points $s=(x,y,z)\in R$.  Note that the vertices of 
$\STAIR(R)$ include all the minimal points of $R$ (and possibly extra
points not in $R$); this orthogonal polyhedron $\STAIR(R)$ has $O(|R|)$
number of vertices (by Euler's formula) and 
can be computed in time $O(|R|\lg|R|)$  
by adapting a standard algorithm for 3-d minima~\cite{GaBeTa,PreShaBOOK}
(the time bound can be improved on the RAM).  
Let $\VD(R)$ denote the cells in the
\emph{vertical decomposition}
of the region underneath $\STAIR(R)$.  The decomposition is defined
as follows: take each horizontal face (a polygon) of $\STAIR(R)$
and form a 2-d vertical decomposition of the face by adding
$y$-vertical line segments at its vertices; finally, extend each resulting
subface (a rectangle) downward to form a cell touching $z=-\infty$.
The decomposition $\VD(R)$ has $O(|R|)$ size and can be computed in
$O(|R|)$ additional time.  For each cell $\D\in\VD(R)$, we define
its \emph{conflict list} $S_\D$ to consist of
all points $s\in S$ with $O_s$ intersecting $\D$; equivalently,
$S_\D$ consists of all points in $S$ that are dominated by the
top-upper-right corner $v_\D$ of $\D$.
The decomposition $\VD(R)$ and its conflict lists, which together
we refer to as a {\em randomized shallow cutting\/} of $S$, satisfy 
some desirable properties\SHORT{ (see the appendix for background on
deterministic shallow cuttings)}:

\begin{lemma}\label{lem:cut}
For a random sample $R$ of $S$ with $\Ex[|R|]=n/K$,
\begin{enumerate}
\item[\rm (a)] $\max_{\D\in\VD(R)}|S_\D| \:=\: O(K\lg N)$ with probability
at least $1-1/N$ for any $N\ge n$;
\item[\rm (b)] $\Ex\left[\sum_{\D\in\VD(R)}|S_\D|\right]\:=\: O(n)$;
\item[\rm (c)] if a query point $q$ dominates exactly $k$ points of $S$, then
$q$ is covered by $\VD(R)$ (i.e., lies below $\STAIR(R)$) with probability at 
least $1-k/K$.
\end{enumerate}
\end{lemma}
\begin{proof}
(a) and (b) follow from the general probabilistic
results by Clarkson and Shor~\cite{ClaDCG87,ClaSho};
(c) is obvious by a union bound: if $q$ is above $\STAIR(R)$, then some 
point of $S$ dominated by $q$ must be chosen in $R$.
\end{proof}

\LONG{
Matou\v sek~\cite{MatCGTA92} provided a 
deterministic version of shallow cuttings
satisfying similar, slightly stronger properties (without the extra logarithmic
factor in (a) and with probability 1 in (c) for $k\le K$).
Originally, shallow cuttings were developed for halfspace range reporting
and defined in terms of arrangements of planes rather than orthants;
the first application to dominance range reporting
was proposed by Afshani~\cite{afshani:dominance}.  
(A similar concept specific to the case of dominance called
{\em $t$-approximate boundary} had also appeared
\cite{firstattempt,Nekrich.SOCG07}.)

For our offline problem, however, preprocessing cost matters and the
above simpler randomized version is more suitable than its deterministic
counterpart. Note that it is not advisable to use shallow cuttings in 4-d
directly, since the number of vertices in the staircase polyhedron $\STAIR(R)$
can be quadratic in $|R|$ in 4-d.
}

\subsection{Offline 3-d Dominance Reporting}
\label{sec:off3d}

We warm up by illustrating how randomized shallow cuttings can help
solve the offline dominance reporting problem in the 3-d case. 
\LONG{The derived solution playes a key role in our 4-d solution.}
We assume that
the given $n$ input points and $n$ query points have been pre-sorted.

\paragraph{Algorithm.}
We pick a random sample $R$ of the input points, where each point is
sampled with probability $1/K$ with $K:=\lg n$. We first compute $\STAIR(R)$ and
$\VD(R)$.

We next compute the conflict lists for all the cells of $\VD(R)$
as follows.  For each input point $s$, it suffices to identify
all cells whose conflict lists include $s$.  We first find the 
cell $\D\in\VD(R)$ containing $s$; this reduces
to a 2-d point location query in the $xy$-projection of $\VD(R)$.
The top-upper-right corner $v_\D$ of $\D$ gives us an initial vertex that
dominates $s$.
We observe that all vertices of the polyhedron $\STAIR(R)$
that dominate the point $s$ form a connected subgraph in the graph
(the 1-skeleton) induced by the polyhedron. Furthermore, the degree of 
each node in the graph is at most 3.
Thus we can perform a breadth-first search from the initial vertex
found to generate all vertices of $\STAIR(R)$
dominating $v$, yielding all conflict lists that include $v$.
The total time over all input points $v$, excluding the initial
point location queries, is
linear in the total size of all conflict lists.

For each query point $q$, we find the cell of $\VD(R)$ 
containing $q$; this again
reduces to a 2-d point location query in the projection of $\VD(R)$.
If no cell is found (i.e., $q$ is above $\STAIR(R)$), then
we say that $q$ is
\emph{bad}; otherwise it is \emph{good}.
For each cell $\D\in\VD(R)$, we run an existing algorithm $A_0$
to solve the offline 3-d dominance reporting subproblem for the input
points in the conflict list of $\D$ and the query points inside~$\D$.

This answers all good queries correctly.  To finish, we
recursively solve the offline 3-d dominance reporting
problem on the bad queries and all the input points,
where the roles of queries and input points are now reversed. Note that we also
reverse the dominance relation, or equivalently, negate all
coordinates.
After recursing twice, however, we terminate by switching to 
a known $O(n\lg n + k)$-time algorithm (e.g.,~\cite{LeePre,PreShaBOOK}).

\newcommand{\Qpl}{Q_{\mbox{\scriptsize\sc pl}}}

\paragraph{Analysis.}
Assume that the offline 2-d point location on $n$ 
pre-sorted rectangles and query points
takes $O(n \Qpl(n))$ time for some non-decreasing function $\Qpl(\cdot)$.
Assume that the initial algorithm $A_0$ for offline 3-d dominance
reporting on $n$ pre-sorted input and query points 
takes $O(nQ_0(n)+k)$ (expected) time for some non-decreasing function
$Q_0(\cdot)$.  
Note that this implies that the running time for
$n$ input points and $m$ query points is $O((n+m)Q_0(n)+k)$, by dividing the
query points into $\lceil m/n\rceil$ groups of size at most~$n$ when $m>n$.

Our algorithm
spends expected time at most $O((n/K)\lg n)=O(n)$ to compute $\STAIR(R)$ and $\VD(R)$.
Performing point locations 
on the $xy$-projection of $\VD(R)$ (a subdivision of expected size $O(n/K)$) for both the input and query points
takes time at most $O(n\Qpl(n))$.  
Observe
here that the input and query points have been pre-sorted,
and so the rectangles can also be pre-sorted in linear time.
Constructing the conflict lists by the breadth-first searches 
takes expected time $O(n)$ since they have expected size $O(n)$ 
by Lemma~\ref{lem:cut}(b). Also by Lemma~\ref{lem:cut}(a),
every conflict
list has size $O(K\lg n)$ w.h.p.; if this condition is violated,
we can afford to switch to a trivial polynomial upper bound on the running time.
Since the total expected size of the 3-d dominance reporting
subproblems at the cells is $O(n)$,
these subproblems can be solved in total expected time $O(nQ_0(O(K\lg n)))
=O(nQ_0(O(\lg^2n)))$ plus the output size.  

One technicality arises: by our assumption, the coordinates 
of the input and query points in each subproblem should be pre-sorted first.
For the $x$-coordinates, this can be accomplished by
scanning through the global sorted $x$-list, and for each
input or query point $s$ in order, appending $s$ to the end of
the linked lists for the cells $s$ participates in.
The $y$- and $z$-sorted lists can be similarly dealt with.  The time required
is linear.

By Lemma~\ref{lem:cut}(c), the probability that a query with 
output size $k_i$ is bad is at most $k_i/K$.  Thus, the expected
number of bad queries is at most $k/K$ for total output size $k$.
After recursing twice, the expected number of
queries and input points both decrease to $O(k/K)$.
The $O(n\lg n+k)$ algorithm would then finish in 
expected time $O((k/K)\lg n +k)=O(k)$.
We conclude that
our algorithm runs in overall expected time $O(n[\Qpl(n)+Q_0(O(\lg^2 n))]+k)$.

\IGNORE{
 When $k=\Omega(n\lg n)$ this is not a problem,
so assume $k=o(n\lg n)$. In this case, there are at least $3n/4$
queries reporting at most $4 \lg n$ points each. Define such a query
to be \emph{shallow}. There must be at least $n/4$ bad shallow queries. 
By Lemma~\ref{lem:cut}(c) (with $\alpha=(4\lg n)/K$), the
probability that a shallow query is bad is at most $(4\lg n)/K$. 
Therefore, the expected number of bad shallow queries is
at most $(4n\lg n)/K$. By Markov's inequality, the probability there
are at least $n/4$ bad shallow queries is bounded by $O((\lg
n)/K)$. This case thus contributes $O(n\lg n \cdot (\lg
n)/K)=O(n)$ to the expected running time of our algorithm.

For the recursive calls, observe that after two calls, the number of
queries and input points have both been halved, thus the time spent at
each recursive level is geometrically decreasing. 
}

For example, we can use $\Qpl(n)=O(\lg\lg n)$ by
the point location method from~\cite{chan_pps} 
(actually in the offline setting, we can just use a plane sweep
algorithm with a dynamic van Emde Boas trees), and $Q_0(n)=O(\lg n)$ by
a known method for 3-d offline dominance reporting~\cite{GupSCG95}. 
Then our algorithm would run in
expected time $O(n\lg\lg n + k)$.  We now show that an even better result is
possible when $n$ is small.

\paragraph{The case of few points.}
First consider the case $n\le w^{O(1)}$. By Lemma~\ref{lem:offpl}, 
$\Qpl(n)=O(1)$.  
We can solve 3-d dominance reporting for $O(\lg^2n)=o(w/\lg w)$
points in linear time, since after rank space reduction, 
the input can be packed into $o(w)$ bits and the answer can
deduced from a word operation.  
Thus, $Q_0(O(\lg^2 n))=O(1)$.
We therefore get an $O(n+k)$-time algorithm.

Next consider the case $n\le 2^{O(\sqrt{w})}$. 
By Lemma~\ref{lem:offpl}, $\Qpl(n)=O(1)$. 
Since $\lg^2 n \le w^{O(1)}$,
by bootstrapping with the first case, we can set $Q_0(O(\lg^2 n))=O(1)$.
We therefore get an $O(n+k)$-time algorithm.

\begin{theorem}
\label{thm:off3d} 
There is an algorithm for offline 3-d dominance reporting on $n$
input points and $n$ query points that runs in expected time
$O(n\lg\lg n + k)$ if the coordinates have been pre-sorted.
The time bound improves to 
$O(n+k)$ if in addition, $n\le 2^{O(\sqrt{w})}$.
\end{theorem}

\subsection{Offline 4-d Dominance Reporting}
\label{sec:off4d}

We are now ready to present our offline 4-d algorithm.
Our algorithm follows the same basic approach employed by most data
structural upper bounds for orthogonal range searching: we construct a
range tree on the input points and solve an offline 3-d problem in
each node of the tree.  Naively, using the $O(n\lg\lg n+k)$ 
algorithm from Theorem~\ref{thm:off3d} would imply
only an $O(n\lg n\lg\lg n + k)$ algorithm (which nevertheless is
an improvement over previous results).   We need several additional
ideas to achieve the final $O(n\lg n + k)$ result.

\paragraph{Algorithm.}
Construct a complete binary tree (range tree) $\TT$ using the input
points ordered by their last coordinate as leaves. Associate each
query point to the leaf node containing its successor input point
w.r.t.\ the last coordinate, and project all input and query points on to
the first three dimensions. Each internal node $u$ in $\TT$ naturally
defines an offline 3-d dominance reporting problem, using the query
points in the right subtree as queries (the query points of $u$), and
the input points in the left subtree as input (the input points of
$u$). Clearly the combined output of all these 3-d problems constitutes
the output for the 4-d problem.

To speed up the solution of these 3-d problems, our first idea is to
use randomized shallow cuttings once again, but this time with a
different choice of parameter $K$.  Pick a random sample of all the
$n$ input points, where each point is included with probability $1/K$
with $K := 2^{\sqrt{w}}$.  For each node $u$ in $\TT$, let $R_u$
denote the sample of the input points of $u$.  We first compute
$\STAIR(R_u)$ and $\VD(R_u)$.  We next compute the conflict lists for
all the cells of $\VD(R_u)$ as in Section~\ref{sec:off3d}%
\LONG{: namely, we
find the cell of $\VD(R_u)$ containing each input point of $u$ by point
location, and then use breadth-first searches to generate the conflict
lists in time linear in their total size.
}\SHORT{
by doing point location in $\VD(R_u)$ for each input point of $u$
and using breadth-first searches.
}
For each query point $q$ of
$u$, we find the cell of $\VD(R_u)$ containing $q$ by point location.
If for a query $q$, there is at least one ancestor node where $q$ is a
query point and no cell is found, we say that $q$ is \emph{bad}.  For
each cell $\D\in\VD(R_u)$, we run the algorithm from
Section~\ref{sec:off3d} to solve the offline 3-d dominance reporting
subproblem for the input points of $u$ in the conflict list of $\D$
and the query points of $u$ inside $\D$ which are not bad.

This answers all queries that are not bad in any node. To finish, we
recursively solve the offline 4-d dominance reporting problem on query
points that are bad in at least one node and all the input points,
where the roles of queries and input points are now reversed.  After
recursing twice, we terminate by switching to a known $O(n\lg^2 n +
k)$-time algorithm (e.g.,~\cite{LeePre,PreShaBOOK,VaiWoo}).

\paragraph{Analysis, excluding point location.}
Our algorithm
spends expected time at most $O((n/K)\lg n)=o(n)$
to compute $\STAIR(R_u)$ and $\VD(R_u)$ per level of the tree.
The breadth-first searches 
take expected time $O(n)$ per level. By Lemma~\ref{lem:cut}(a),
every conflict
list has size $O(K\lg n)=2^{O(\sqrt{w})}$ w.h.p.; if this condition is violated at any node,
we can afford to switch to a trivial polynomial upper bound on the running time.
Since the total expected size of the 3-d dominance reporting
subproblems at the cells is $O(n)$,
these subproblems can be solved in total expected time $O(n)$ per level, plus
the output size, by applying Theorem~\ref{thm:off3d} in the ``few points'' case.  
One technicality arises: the coordinates 
of the input and query points in each subproblem should be pre-sorted first.
As in Section~\ref{sec:off3d}, this can be accomplished by
scanning through the global sorted $x$-, $y$-, and $z$-lists in linear time.
The total time excluding point location cost is then
$O(n)$ per level, i.e., $O(n\log n)$, plus the output size.

By Lemma~\ref{lem:cut}(c), the probability that a query with output
size $k_i$ is bad at one or more nodes is at most $k_i/K$.  Thus, the
expected total number of bad queries at all nodes is at most $k/K$ for
total output size $k$.  After recursing twice, the expected number of
queries and input points both decrease to $O(k/K)$.  The $O(n\lg^2
n+k)$ algorithm would then finish in expected time $O((k/K)\lg^2 n
+k)=O(k)$.

\paragraph{Point location cost.}
At each node $u$, we need to perform point locations on 
the $xy$-projection of $\VD(R_u)$ for
all input points and query points of $u$.
Unfortunately, the current best offline 2-d orthogonal point location algorithm in general
requires $O(\lg\lg n)$ time per query, which would result in suboptimal
total time $O(n\lg n\lg\lg n)$.  We suggest the following key idea: solve
all the 2-d point location subproblems collectively, by transforming them
into one single 3-d problem!

Specifically, consider point locations for the query points of $u$;
locations of the input points of $u$ can be dealt with similarly.
The query points for which we must perform a
point location in $\VD(R_u)$ are precisely those in the right subtree of
$u$. These queries lie in a consecutive range of leaves, say
$\ell_i$ through $\ell_j$, counted from left to right. We now
transform each rectangle $r=[x_1,x_2] \times [y_1,y_2]$ in the
$xy$-projection of $\VD(R_u)$ into
the 3-d rectangle $r'=[x_1,x_2] \times [y_1,y_2] \times [i,j]$ and
collect the set $B$ of all such 3-d boxes over all nodes in
$\TT$. Similarly, we transform each query point $q$ to another 3-d query
point $q'$. If $q$ has coordinates $(x,y,z)$ and lies in leaf
$\ell_i$, we map $q$ to the point $q'=(x,y,i)$. We then collect the
set $A$ of all transformed query points, and solve an offline 3-d
orthogonal range reporting problem with the points in $A$ and the boxes in $B$.
From the output of this offline problem,
we can obtain for each query point in $A$ the set of boxes in $B$ that contain
it. This gives the answers to all the original 2-d point location queries.

Since the subdivisions $\VD(R_u)$ have total expected size $O(n/K)$ per level
of the tree, the expected number of boxes in $B$ is $O((n/K)\lg n)$.  
On the other hand, the number of points in $A$ is $n$, and 
the total output size of the 3-d problem is $O(n\lg n)$,
since each point in $A$ lies in $O(\lg n)$ boxes in $B$.
By applying Lemma~\ref{lem:offstab} with $b=K^\eps$, we can solve the offline
3-d orthogonal range reporting problem in expected time
$$O(n\log_b^2 n + b^2 (n/K)\lg n\log_b^2 n + n\lg n)\:=\:
O(n(\lg n/\lg K)^2 + n\lg n)\:=\:O(n\lg n),$$
due to the fortuitous choice of $K=2^{\sqrt{w}}$.
We finally conclude

\begin{theorem}\label{thm:off4d}
There is an algorithm for offline 4-d dominance reporting on $n$
input points and $n$ query points that runs in expected time $O(n \lg
n+k)$, where $k$ is the total output size.
\end{theorem}

\newcommand{\OFFLINEAPPENDIX}{
\subsection{Remarks\SHORT{ on Section~\ref{sec:offline}}}
\label{sec:offline:rmks}

\paragraph{Background on shallow cuttings.}
Matou\v sek~\cite{MatCGTA92} provided a 
deterministic version of shallow cuttings
satisfying similar, slightly stronger properties as
in Lemma~\ref{lem:cut} (without the extra logarithmic
factor in (a) and with probability 1 in (c) for $k\le K$).
Originally, shallow cuttings were developed for halfspace range reporting
and defined in terms of arrangements of planes rather than orthants;
the first application to dominance range reporting
was proposed by Afshani~\cite{afshani:dominance}.  
For our offline problem, however, preprocessing cost matters and the
above simpler randomized version is more suitable than its deterministic
counterpart. Note that it is not advisable to use shallow cuttings in 4-d
directly, since the number of vertices in the staircase polyhedron $\STAIR(R)$
can be quadratic in $|R|$ in 4-d.

\paragraph{Higher dimensions and applications.}

The $d$-dimensional problem reduces to the $(d-1)$-dimensional problem
at the expense of a logarithmic factor increase, by standard divide-and-conquer.
Theorem~\ref{thm:off4d} thus implies:

\begin{corollary}
There is an algorithm for offline $d$-dimensional dominance reporting on $n$
input points and $n$ query points that runs in expected time $O(n \lg^{d-3}
n+k)$ for any constant $d\ge 4$, where $k$ is the total output size.  
\end{corollary}

Our method also works for offline dominance emptiness; the same
bounds hold with $k$ set to 0.
In fact, the algorithms can be slightly simplified: a query point
that dominates no input points is good with probability~1 by 
Lemma~\ref{lem:cut}(c), and so there is no need to recurse on the
bad queries.  

The problem of reporting enclosure pairs for $d$-dimensional boxes immediately
reduces to $(2d)$-dimensional dominance reporting.
The $d$-dimensional maxima problem obviously reduces to answering 
$n$ $d$-dimensional offline dominance emptiness queries.

For a less obvious application, consider the computation of the
{\em $L_\infty$-minimum spanning tree\/} of $n$ points.
A reduction by Krznaric, Levcopoulos, and Nilsson~\cite{KrLeNi}
showed that this problem can be reduced to
the {\em bichromatic $L_\infty$-closest pair\/} problem:
given $n$ red points and $n$ blue points, find a pair of red and
blue points with the smallest $L_\infty$-distance. 
Their reduction does not increase the asymptotic running time, 
if it exceeds $n\lg n$.
A randomized optimization technique by Chan~\cite{ChaSCG98}
showed that the problem can be further reduced to
the following decision problem, without increasing the
asymptotic expected running time:
given $n$ red points and $n$ blue points and a value $r$,
decide whether the $L_\infty$-distance is at most $r$.
By drawing hypercubes centered at the blue points of side length $2r$,
this problem in turn is equivalent to deciding whether some blue hypercube
contains some red point.  
Build a grid of side length $r$.
We can assign points and hypercubes
to grid cells via hashing in linear expected time.
Inside each cell,
the blue hypercubes are $d$-sided.  So, the problem reduces
to a collection of offline dominance emptiness subproblems with
linear total size.

\begin{corollary}
The maxima problem, the bichromatic $L_\infty$-closest pair
problem, and the $L_\infty$-minimum spanning tree problem
in any constant dimension $d\ge 4$ 
can be solved in $O(n\lg^{d-3}n)$ expected time.
\end{corollary}
}
\LONG{\OFFLINEAPPENDIX}

{\small
\bibliographystyle{abbrv}
\bibliography{ors}
}

\appendix

\section{Appendix}

\SHORT{\THREEDAPPENDIX\OFFLINEAPPENDIX}

\subsection{Succinct Rank Queries (Proof of Lemma~\ref{lem:succinct})}
The proof is rather standard. We will store a ``checkpoint'' once
every $\Sigma \lg n$ positions in the array: a record with $\Sigma$
entries (of $\lg n$ bits each) that indicates how many elements of
each kind we have prior to that position. Then, for each element in
the array, we can simply write $A[i]$ and the number of elements equal
to $A[i]$ between the last checkpoint and $i$. This uses $O(\lg
(\Sigma\lg n))$ bits per element, so it fits our space bound if
$\Sigma \ge \sqrt{\lg n}$. The query simply adds the counter stored
with $A[i]$ and the appropriate counter from the last checkpoint.

If the alphabet is smaller, we employ a 2-level scheme. We store a
checkpoint as above every $\Sigma \lg n$ positions. Additionally, every
$\Sigma \lg\lg n$ positions, we store a minor checkpoint: a record of
$\Sigma \lg (\Sigma \lg n) = O(\Sigma \lg\lg n)$ bits which indicates
the number of elements of each kind from the last checkpoint to the
minor checkpoint. A query retrieves the appropriate counters from the
last checkpoint and the last minor checkpoint, and then must solve the
rank problem between the last checkpoint and the query position. Since
$\Sigma \lg\lg n \le \sqrt{\lg n} \cdot \lg\lg n$, the array entries
between minor checkpoints fit in $O(\sqrt{\lg n} \cdot \lg^2 \lg n)$
bits. Thus, we can simply store the array entries in plain form, and
use a precomputed table of space $n^{o(1)}$ to answer rank queries
between minor checkpoints in constant time.

\subsection{Offline 2-d Orthogonal Point Location for Few Points (Proof of Lemma~\ref{lem:offpl})}

The proof follows the bit-packing approach of Chan and \Patrascu~\cite{ChaPatSODA10}.
First scan through the sorted input lists to reduce all coordinates to rank
space, 
that is, every coordinate of a query and
rectangle is an integer of value $O(n)$. 
\IGNORE{
Note that by using integer
coordinates of value at most $4n$, we can easily preserve the
containment relation. If $m=\Omega(n^2)$, we furthermore pick one
representative query point for each possible location of a query in
$[4n]^2$. We solve the point location problem on the set of $O(\min
\{m,n^2\})$ representative queries and the set of $n$ rectangles. From
the output of this smaller problem, we immediately get the output to
the full problem since queries having the same coordinates lie in the
same rectangle. The purpose of selecting representatives is that we
can give each representative a unique id of $O(\lg n)$ bits, instead
of $O(\lg m)$ bits. Now that each query has an id of $O(\lg n)$ bits,
and all coordinates are represented by $O(\lg n)$ bits, we create two
lists $R$ and $Q$ containing the rectangles and the representative
queries respectively. These lists are sorted by $y$-coordinates. In
the representation of $R$ and $Q$, we have packed each group of
$O(w/\lg n)$ consecutive elements into one word, and in $R$, the
rectangle are ordered by their bottom side.
}

We reduce our problem to a number of 1-d disjoint-intervals stabbing
problems (given a set of points and disjoint intervals, return for
each point the interval containing it if the interval exists). 
Essentially we construct a segment tree on the
rectangles (where we divide according to $x$-coordinates) and solve a 1-d problem in each node: Consider a trie of
depth $O(\lg n)$ over the binary alphabet. For each rectangle
$r=[x_1,x_2] \times [y_1,y_2]$, let $\ell_1$ and $\ell_2$ denote
the leaves corresponding to the binary representation of $x_1$ and
$x_2$. Now consider the two paths from these leaves to their lowest
common ancestor. For each node $u$ on the path to $\ell_1$ where
$\ell_1$ lies in the left child's subtree, associate the interval
$[y_1,y_2]$ to the right child. For $\ell_2$, do the same, but with
the roles of left and right reversed. Observe that the $y$-intervals
associated with each node in the trie are disjoint, by the
disjointness of the rectangles. Our first task is to compute
for each node in the trie, a sorted list of the associated $y$-intervals,
where the list has been packed into words to allow $O(w/\lg n)$
consecutive intervals to be stored in one word.

We construct these lists essentially by external-memory radix
sorting. We start at the root node where we are given the complete
input set $S$ in sorted order of bottom $y$-coordinates. We scan over
this list and distribute the rectangles to two sets, $S_\ell$ and
$S_r$, one for the left child of the root, and one for the
right. These lists are again packed into words. The set $S_\ell$
contains those rectangles $[x_1,x_2] \times [y_1,y_2]$ for which the
left child lies on the path from the root to either of the leaves
corresponding to the binary representation of $x_1$ or
$x_2$. The set $S_r$ is similar. Observe that this can be determined
directly from the binary representation of $x_1$ and
$x_2$. Furthermore, if $[x_1,x_2]$ completely contains the range of
$x$-coordinates stored in the leaves associated with either the left
or the right subtree, we append $[y_1,y_2]$ to a list stored for the
root of that subtree. These lists are also packed into words. Finally,
we recurse on the left and right subtree, using $S_\ell$ and $S_r$ as
input respectively.

Since the rectangles in $S$ are given in sorted $y$-order and 
the $y$-intervals associated with each node are
disjoint, this correctly constructs the desired lists. Furthermore,
observe that we can handle all $O(w/\lg n)$ rectangles stored in one
word in $O(1)$ time using table lookups. Thus we spend $O((n\lg n)/w)$
time on each of $O(\lg n)$ levels of the trie to construct the desired
lists; the total time is $O((n\lg^2n)/w)$.

We now associate each query $(x,y)$ to the set of nodes on the path
from the root to the leaf corresponding to the binary representation
of $x$. Using the same approach as for the rectangles, we obtain a
$y$-sorted list of the associated queries in each node of the trie in
total time $O((n\lg^2n)/w)$. To finish, we solve the 1-d disjoint-intervals 
stabbing problem in each node of the trie by scanning the
list of associated intervals and the list of associated queries
simultaneously (in order of $y$-coordinates). Note that we produce
output only when an interval contains a query point. Using table
lookups when performing the scan, we may advance at least one word in
one of the lists in $O(1+k')$ time, where $k'$ denotes the output size
between the queries and the intervals in the two considered
words. Over the entire trie, the total output size is $O(n)$
and the total cost is $O((n\lg^2n)/w+n)=O(n)$ for $n\le 2^{O(\sqrt{w})}$.

\subsection{An Alternative Algorithm for a Special Case of 4-d Offline
Dominance Emptiness}

In this subsection, we give an alternative 
\emph{deterministic} $O(n\lg n)$-time algorithm for a special case
of 4-d offline dominance emptiness: given $n$ red points and
$n$ blue points, decide whether there exists a red point $p$
and a blue point $q$ such that $p$ is dominated by $q$.
The original 4-d offline dominance emptiness problem is stronger:
there, we want to know for every blue point $q$ whether there exists
a red point dominated by $q$.  This special case is sufficient,
for example, to solve the bichromatic $L_\infty$-closest pair 
and the $L_\infty$-minimum spanning tree problem discussed
in Section~\ref{sec:offline:rmks}; however, it is not sufficient
to solve the maxima problem.
The alternative algorithm has the advantage that it avoids ``bit tricks'',
though it requires a nontrivial subroutine---an algorithm
of Chazelle~\cite{ChaSICOMP92} for intersecting convex polyhedra.

We first consider the problem in 3-d.  It has been observed that
techniques for halfspace range searching can often be adapted
to dominance range searching~\cite{afshani:dominance}.  We first point out
an explicit way to reduce dominance to halfspace range searching.
Surprisingly, this reduction has not appeared before to the authors'
knowledge (although
the idea behind the reduction, which is based on an exponentially
spaced grid, is commonplace).
Specifically, fix a constant $r>3$.
Assume that all the points have positive integer coordinates (we can initially
sort the coordinates once at the beginning and reduce to rank space).
Transform each red point $p=(i,j,k)$ to the point
$p^*=(r^i,r^j,r^k)$.
Transform each blue point $q=(a,b,c)$ to the halfspace
$q^*=\{(x,y,z) : x/r^a + y/r^b + z/r^c \le 3\}$.
It is easy to see that $p$ is dominated by $q$ iff $p^*$ lies inside $q^*$:
if $i\le a$, $j\le b$, and $k\le c$, then $r^i/r^a+r^j/r^b+r^k/r^c\le 3$,
but if $i>a$, $j>b$, or $k>c$, then $r^i/r^a+r^j/r^b+r^k/r^c \ge r > 3$.

Let $P$ be the convex hull of the transformed red points and $Q$ be the 
intersection of the complements of the transformed blue halfspaces.
Then the answer to our red/blue dominance problem is no iff
every transformed red point lies in the complement of every transformed
blue halfspace, i.e., $P$ lies inside $Q$, i.e., $P\cap Q = P$.
Chazelle~\cite{ChaSICOMP92} has given a linear-time algorithm
for intersecting two convex polyhedra.  Thus, the 3-d problem
can be solved in linear time, provided that the
polyhedra $P$ and $Q$ are given.

Now, to solve the 4-d problem, we build a binary range tree $\TT$ according to
the last coordinate as before and obtain a series of 3-d subproblems of total
size $O(n\lg n)$.  Observe that we can pre-compute the
red convex hulls $P$ at all the nodes of $\TT$ bottom-up in $O(n\lg n)$ time, by
repeatedly using Chazelle's linear-time algorithm for merging
two convex hulls (computing the convex hull of two convex polyhedra
is dual to intersecting two convex polyhedra).  Similarly, we can
pre-compute the blue halfspace intersections $Q$ at all the nodes
of $\TT$ in $O(n\lg n)$ time, again 
by repeatedly using Chazelle's algorithm
for intersecting two halfspace intersections.  
This immediately gives a solution to the 4-d problem
with overall running time $O(n\lg n)$.

\medskip\noindent
{\em Remarks\/}: Precision issues seem to arise since the coordinates of
the transformed points and halfspaces involve large numbers, but 
we can simulate any primitive operation on these points and halfspaces
by treating $r$ as a symbolic variable that approaches infinity.  It would 
be interesting to see if we can directly merge or intersect staircase polyhedra 
without going through the transformation and invoking Chazelle's
algorithm.  Note that the above approach does not work at all for the offline
dominance reporting problem, or for that matter, the
offline dominance emptiness problem (in the 3-d subproblem,
we do not know the answer for any non-maximal blue query point whose halfspace 
does not appear on $\partial Q$).

\end{document}